\documentclass{article}

\usepackage[final]{nips_2016}
\pdfoutput=1
\usepackage[utf8]{inputenc} 
\usepackage[T1]{fontenc}    
\usepackage{url}            
\usepackage{booktabs}       
\usepackage{amsfonts}       
\usepackage{nicefrac}       
\usepackage{microtype}      
\usepackage{amsmath,amsfonts,amsthm,amssymb}
\usepackage{algorithm}
\usepackage{algorithmic}
\usepackage{times}
\usepackage{natbib}
\usepackage{graphicx} 
\usepackage{subfigure} 

\newtheorem{theorem}{Theorem}[section]
\newtheorem{proposition}[theorem]{Proposition}
\newtheorem{corollary}[theorem]{Corollary}

\newtheorem{lemma}{Lemma}[section]

\newcommand{\leaveout}[1]{}

\title{A General Retraining Framework for Scalable Adversarial Classification}

\author{
Bo Li \\
\texttt{bo.li.2@vanderbilt.edu} \\
Vanderbilt University \\
\And
 Yevgeniy Vorobeychik   \\
\texttt{yevgeniy.vorobeychik@vanderbilt.edu} \\
Vanderbilt University \\
\And
 Xinyun Chen  \\
\texttt{jungyhuk@gmail.com}          \\
Shanghai Jiaotong University \\
}

\begin{document}

\maketitle

\begin{abstract}
Traditional classification algorithms assume that training and test
data come from similar distributions.
This assumption is violated in adversarial settings,
where malicious actors modify instances to evade detection.
A number of 
custom methods have been developed for both adversarial evasion attacks and 
robust learning.
We propose the first systematic and general-purpose retraining
framework which can: a) boost robustness of an \emph{arbitrary} learning
algorithm, in the face of b) a broader class of adversarial models
than any prior methods.
We show that, under natural
conditions, the retraining framework minimizes an upper
bound on optimal adversarial risk, and show how to extend this result
to account for approximations of evasion attacks.
Extensive experimental evaluation demonstrates that our retraining
methods are nearly indistinguishable from state-of-the-art algorithms
for optimizing adversarial risk, but are \emph{more general
and far more scalable}.
The experiments also confirm that without retraining, our adversarial
framework dramatically reduces the effectiveness of learning.
In contrast, retraining significantly boosts
robustness to evasion attacks without significantly compromising 
overall accuracy.
\end{abstract}

\section{Introduction}
Machine learning has been used ubiquitously for a wide variety of security tasks, such as intrusion detection, malware detection, spam filtering, and web search~\cite{fawcett1997adaptive,mahoney2002learning,fogla2006polymorphic,huang2011adversarial,kolcz2009feature}. 
Traditional machine learning systems, however, do not account for adversarial manipulation.
For example, in spam detection, spammers commonly change spam email text to evade filtering.
As a consequence, there have been a series of efforts to both model adversarial manipulation of learning, such as evasion and data poisoning attacks~\cite{lowd2005adversarial,karlberger2007exploiting,nelson2012query}, as well as detecting such attacks~\cite{mahoney2002learning,el2003robust} or enhancing robustness of learning algorithms to these~\cite{li2014feature,li2015scalable,teo2007convex,globerson2006nightmare,bruckner2011stackelberg}.
One of the most general of these, due to Li and Vorobeychik~\cite{li2014feature}, admits evasion attacks modeled through a broad class of optimization problems, giving rise to a Stackelberg game, in which the learner minimizes an adversarial risk function which accounts for optimal attacks on the learner.
The main limitation of this approach, however, is scalability: it can solve instances with only 10-30 features.
Indeed, most approaches to date also offer solutions that build on specific learning models or algorithms.
For example, specific evasion attacks have been developed for linear or convex-inducing classifiers~\cite{lowd2005adversarial,karlberger2007exploiting,nelson2012query}, as well as neural networks for continuous feature spaces~\cite{biggio2014security}.
Similarly, robust algorithms have typically involved non-trivial modifications of underlying learning algorithms, and either assume a specific attack model or modify a specific algorithm.
The more general algorithms that admit a wide array of attack models, on the other hand, have poor scalability.

We propose a very general retraining framework, \emph{RAD}, which can boost evasion robustness of arbitrary learning algorithms using arbitrary evasion attack models.
We show that \emph{RAD} minimizes an upper bound on optimal adversarial risk.
This is significant: whereas adversarial risk minimization is a hard bi-level optimization problem and poor scalability properties (indeed, no methods exist to solve it for general attack models), \emph{RAD} itself is extremely scalable in practice, as our experiments show.
We develop \emph{RAD} for a more specific, but very broad class of adversarial models, offering a theoretical connection to adversarial risk minimization even when the adversarial model is only approximate.
In the process, we offer a simple and very general class of local search algorithms for approximating evasion attacks, which are experimentally quite effective.
Perhaps the most appealing aspect of the proposed approach is that \emph{it requires no modification of learning algorithms}: rather, it can wrap any learning algorithm ``out-of-the-box''.
Our work connects to, and systematizes, several previous approaches which used training with adversarial examples to either evaluate robustness of learning algorithms, or enhance learning robustness.
For example, Goodfellow et al.~\cite{goodfellow2014explaining} and Kantchelian et al.~\cite{Kantchelian15} make use of adversarial examples.
In the former case, however, these were essentially randomly chosen.
The latter offered an iterative retraining approach more in the spirit of \emph{RAD}, but did not systematically develop or analyze it.
Teo et al.~\cite{teo2007convex} do not make an explicit connection to retraining, but suggest equivalence between their general invariance-based approach using column generation and retraining.
However, the two are not equivalent, and Teo et al. did not study their relationship formally.

In summary, we make the following contributions:
\begin{enumerate}
\item \emph{RAD}, a novel systematic framework for adversarial retraining,
\item analysis of the relationship between \emph{RAD} and empirical adversarial risk minimization,
\item extension of the analysis to account for approximate adversarial evasion, within a specific broad class of adversarial models,
\item extensive experimental evaluation of \emph{RAD} and the adversarial evasion model.
\end{enumerate}
We illustrate the applicability and efficiency of our method on both spam filtering and handwritten digit recognition tasks, where evasion attacks are extremely salient~\cite{klimt2004enron,lecun2010mnist}. 

\section{Learning and Evasion Attacks}

Let $X \subseteq R^n$ be the feature space, with $n$ the number of
features. For a feature vector $x_i \in X$, we let $x_{i_{j}}$ denote the
$j$th feature. Suppose that the training set $(x_i,y_i)$ is comprised of
feature vectors $x_i \in X$ generated according to some unknown
distribution $x_i \sim F$, with $y_i \in \{ -1,  +1\}$ the corresponding
binary labels, where $-1$ means the instance $x_i$ is benign, while $+1$
indicates a malicious instance. 
The learner aims to learn a classifier with parameters $\beta$,
$g_\beta: X \to \{-1, +1\}$, to label instances as malicious or
benign, using a training data set of labeled instance $D = \{(x_1, y_1),
..., (x_m, y_m)\}$. 
Let $I_{bad}$ be the subset of datapoints $i$ with $y_i = +1$.
We assume that $g_\beta(x) = \mathrm{sgn}(f_\beta(x))$ for some real-valued
function $f_\beta(x)$.






Traditionally, machine learning algorithms commonly minimize regularized empirical
risk:
\begin{equation}
\label{E:standardloss}
\min_\beta \mathcal{L}(\beta) \equiv \sum_i l(g_\beta(x_i),y_i) + \alpha \|\beta\|_p^p,
\end{equation}
where $l(\hat{y},y)$ is the loss associated with predicting $\hat{y}$
when true classification is $y$.
An important issue in adversarial settings is that instances
classified as malicious (in our convention, corresponding to
$g_\beta(x) = +1$) are associated with malicious agents who
subsequently modify such instances in order to evade the classifier
(and be classified as benign).
Suppose that adversarial evasion behavior is captured by an oracle,
$\mathcal{O}(\beta, x)$, which returns, for a given parameter vector
$\beta$ and original feature vector (in the training data) $x$, an
alternative feature vector $x'$.
Since the adversary modifies malicious instances according to this
oracle, the resulting effective risk for the defender is no longer
captured by Equation~\ref{E:standardloss}, but must account for
adversarial response.
Consequently, the defender would seek to minimize the following
\emph{adversarial risk} (on training data):
\begin{align}
\label{E:advloss}
\min_\beta \mathcal{L}_A(\beta; \mathcal{O}) = \sum_{i:y_i = -1}
l(g_\beta(x_i),-1)+ \sum_{i:y_i = +1} l(g_\beta(\mathcal{O}(\beta,x_i),+1) +\alpha \|\beta\|_p^p.
\end{align}
The adversarial risk function in Equation~\ref{E:advloss} is extremely
general: we make, at the moment, no assumptions on the nature of the
attacker oracle, $\mathcal{O}$.
This oracle may capture evasion attack models based on minimizing
evasion
cost~\cite{lowd2005adversarial,li2014feature,biggio2014security}, or
based on actual attacker evasion behavior obtained from experimental data~\cite{ke2016behavioral}.

\section{Adversarial Learning through Retraining}

A number of approaches have been proposed for making learning
algorithms more robust to adversarial evasion
attacks~\cite{dalvi2004adversarial,li2014feature,li2015scalable,teo2007convex,bruckner2011stackelberg}.
However, these approaches typically suffer from three limitations: 1)
they usually assume specific attack models, 2) they require
substantial modifications of learning algorithms, and 3) they commonly
suffer from significant scalability limitations.
For example, a recent, general, adversarial learning algorithm
proposed by Li and Vorobeychik~\cite{li2014feature} makes use of
constraint generation, but does not scale beyond 10-30 features.

Recently, retraining with adversarial data has been proposed as a means to increase
robustness of
learning~\cite{goodfellow2014explaining,Kantchelian15,teo2007convex}.\footnote{Indeed,
neither Teo et al.~\cite{teo2007convex} nor Kantchelian et
al.~\cite{Kantchelian15} focus on retraining as a main contribution,
but observe its effectiveness.}
However, to date such approaches have not been systematic.

We present a new algorithm, \emph{RAD}, for retraining with adversarial data
(Algorithm~\ref{algo:retraining}) which systematizes some of
the prior insights, and enables us to provide a formal connection
between retraining with adversarial data, and adversarial risk
minimization in the sense of Equation~\ref{E:advloss}.
\begin{algorithm}[thb]
\caption{RAD: Retraining with ADversarial Examples}
\label{algo:retraining}
\begin{algorithmic}[1]
\STATE \textbf{Input}: {training data $X$}
\STATE $N_i \gets \emptyset \ \forall \ i \in I_{bad}$
\REPEAT
\STATE $\beta \gets$ Train($X \cup_i N_i$)
\STATE $new \gets \emptyset$
\FOR {$i \in I_{bad}$}
    \STATE $x'=\mathcal{O}(\beta,x_i)$
    \IF{$x' \notin N_i$}
      \STATE $new \gets new \cup x'$
    \ENDIF
    \STATE $N_i \gets N_i \cup x'$
  \ENDFOR
\UNTIL{$new = \emptyset$}
\STATE \textbf{Output}: {Parameter vector $\beta$}
\end{algorithmic}
\end{algorithm}
The \emph{RAD} algorithm is quite general. 
At the high level, it starts with the original training data $X$ and
iterates between computing a classifier and adding adversarial
instances to the training data that evade the previously computed
classifier, if they are not already a part of the data.

A baseline termination condition for \emph{RAD} is that no new
adversarial instances can be added (either because instances generated
by $\mathcal{O}$ have already been previously added, or because the
adversary's can no longer benefit from evasion, as discussed more
formally in Section~\ref{s:game}).
If the range of $\mathcal{O}$ is finite (e.g., if the feature space is
finite), \emph{RAD} with this termination condition would always terminate.
In practice, our experiments demonstrate that when termination
conditions are satisfied, the number of \emph{RAD} iterations is quite
small (between 5 and 20).
Moreover, while \emph{RAD} effective increases the importance of malicious
instances in training, this does not appear to significantly harm
classification performance in a non-adversarial setting.
In general, we can also control the number of rounds directly, or use
an additional termination condition, such as that the parameter vector
$\beta$ changes little between successive iterations.
However, we assume henceforth that there is no fixed iteration limit or convergence check.

To analyze what happens if the algorithm terminates, define the regularized empirical risk in the last
iteration of \emph{RAD} as:
\begin{equation}
\label{E:radloss}
\mathcal{L}^R_N(\beta, \mathcal{O}) = \sum_{i \in D \cup N} l(g_\beta(x_i),y_i) + \alpha ||\beta||_p^p,
\end{equation}
where a set $N = \cup_i N_i$ of data points has been added
by the algorithm (we omit its dependence on $\mathcal{O}$ to simplify
notation).
We now characterize the relationship between $\mathcal{L}^R_N(\beta, \mathcal{O})$
and $\mathcal{L}_A^*(\mathcal{O}) = \min_\beta \mathcal{L}_A(\beta,
\mathcal{O})$.
\begin{proposition}
\label{prop:defender_loss}
$\mathcal{L}_A^*(\mathcal{O}) \le \mathcal{L}^R_N(\beta, \mathcal{O})$
for all $\beta, \mathcal{O}$.
\end{proposition}
\begin{proof}
Let $\bar{\beta} \in \arg\min_{\beta}  \mathcal{L}^R_N(\beta,
\mathcal{O})$.
Consequently, for any $\beta$,
\begin{align*}
\mathcal{L}^R_N(\beta, \mathcal{O}) &\ge \mathcal{L}^R_N(\bar{\beta},
\mathcal{O})\\
&= \sum_{i: y_i = -1} l(g_{\bar{\beta}}(x_i),-1) +\sum_{i:y_i = +1} \sum_{j
  \in N_i \cup x_i} l(g_{\bar{\beta}} (x_i),+1) + \alpha ||\bar{\beta}||_p^p\\
&\ge \sum_{i: y_i = -1} l(g_{\bar{\beta}}(x_i),-1) +\sum_{i:y_i = +1}
l(g_{\bar{\beta}} (\mathcal{O}(\bar{\beta},x_i)),+1) + \alpha ||\bar{\beta}||_p^p\\
&\ge \min_\beta \mathcal{L}_A(\beta; \mathcal{O}) = \mathcal{L}_A^*(\mathcal{O}),
\end{align*}
where the second inequality follows because in the last iteration of the algorithm, $new =
\emptyset$ (since it must terminate after this iteration), which means
that $\mathcal{O}(\beta,x_i) \in N_i$ for all $i \in I_{bad}$.
\end{proof}
In words, retraining, systematized in the \emph{RAD} algorithm,
effectively minimizes an upper bound on optimal adversarial
risk.\footnote{Note that the bound relies on the fact that we are
  only adding adversarial instances, and terminate once no more
  instances can be added.  In particular, natural variations, such as
  removing or re-weighing added adversarial instances to retain original
  malicious-benign balance lose this guarantee.}
This offers a conceptual explanation for the previously observed
effectiveness of such algorithms in boosting robustness of learning to
adversarial evasion.
Formally, however, the result above is limited for several
reasons.
First, for many adversarial models in prior literature, adversarial
evasion is NP-Hard.
While some effective approaches exist to compute optimal evasion for specific
learning algorithms~\cite{Kantchelian15}, this is not true in general.
Although approximation algorithms for these models exist, using them
as oracles in \emph{RAD} is problematic, since actual attackers may
compute better solutions, and Proposition~\ref{prop:defender_loss} no
longer applies.
Second, we assume that $\mathcal{O}$ returns a unique result, but when
evasion is modeled as optimization, optima need not be unique.
Third, there do not exist effective general-purpose adversarial evasion
algorithms the use of which in \emph{RAD} would allow reasonable
theoretical guarantees.
Below, we investigate an important and very general class of
adversarial evasion models and associated algorithms which allow us to
obtain practically meaningful guarantees for \emph{RAD}.

\noindent{\bf Clustering Malicious Instances: }  A significant
enhancement in speed of the approach can be obtained by clustering
malicious instances: this would reduce both the number of iterations,
as well as the number of data points added per iteration. Experiments
(in the supplement) show that this is indeed quite effective.

\noindent{\bf Stochastic gradient descent: } \emph{RAD}
works particularly well with online methods, such as stochastic
gradient descent. Indeed, in this case we need only to make gradient
descent steps for newly added malicious instances, which can be added
one at a time until convergence.
\section{Modeling Attackers}

\subsection{Evasion Attack as Optimization}
\label{s:game}

In prior literature, evasion attacks have almost universally been
modeled as optimization problems in which attackers balance the
objective of evading the classifier (by changing the label from $+1$
to $-1$) and the cost of such
evasion~\cite{lowd2005adversarial,li2014feature}.
Our approach is in the same spirit, but is formally somewhat distinct.
In particular, we assume that the adversary has the following
two competing objectives: 1) appear as benign as possible to the
classifier, and 2) minimize modification cost.
It is also natural to assume that the attacker obtains no value from a
modification to the original feature vector if the result is still
classified as malicious.
To formalize, consider an attacker who in the original training data
uses a feature vector $x_i$ ($i \in I_{bad})$).
The adversary $i$ is solving the following optimization problem:
\begin{equation}
\label{eq:attack}
\min_{x \in X} \min\{0,f(x)\} + c(x,x_i).
\end{equation}
We assume that $c(x,x_i) \ge 0$,
$c(x,x_i) = 0$ iff $x = x_i$, and $c$ is strictly increasing
in $\|x - x_i\|_2$ and strictly convex in $x$.\footnote{Here we
  exhibit a particular general attack model, but many alternatives are
  possible, such as using constrained
  optimization. We found
  experimentally that the results are not particularly
  sensitive to the choice of the attack model.}
Because Problem~\ref{eq:attack} is non-convex, we instead minimize an upper bound:
\begin{equation}
\label{eq:attacker_q}
\min \limits_{x} Q(x) \equiv f(x) + c(x,x_i).
\end{equation}
In addition, if $f(x_i) < 0$, we return $x_i$ before solving Problem~\ref{eq:attacker_q}.
If Problem~\ref{eq:attacker_q} returns an optimal solution $x^*$ with
$f(x^*) \ge 0$, we return $x_i$; otherwise, return $x^*$.
Problem~\ref{eq:attacker_q} has two advantages.
First, if $f(x)$ is convex and $x$ real-valued, this is a (strictly) convex
optimization problem, has a unique solution, and we can solve it in polynomial time.
An important special case is when $f(x) = w^Tx$.
The second one we formalize in the following lemma.
\begin{lemma}
\label{L:bound}
Suppose $x^*$ is the optimal solution to Problem~\ref{eq:attack}, $x_i$ is
suboptimal, and $f(x^*) < 0$. 
Let $\bar{x}$ be the optimal solution to
Problem~\ref{eq:attacker_q}.
Then $f(\bar{x}) + c(\bar{x},x_i) = f(x^*) + c(x^*,x_i)$, and
$f(\bar{x}) < 0$.
\end{lemma}
The following corollary then follows by uniqueness of optimal
solutions for strictly convex objective functions over a real vector space.
\begin{corollary}
\label{C:convex}
If $f(x)$ is convex and $x$ continuous, $x^*$ is the optimal solution to
Problem~\ref{eq:attack}, $\bar{x}$ is the optimal solution to
Problem~\ref{eq:attacker_q}, and $f(x^*) < 0$, then $\bar{x} = x^*$.
\end{corollary}
A direct consequence of this corollary is that when we use
Problem~\ref{eq:attacker_q} to approximate Problem~\ref{eq:attack} and
this approximation is convex, we always return either the optimal
evasion, or $x_i$ if no cost-effective evasion is possible.
An oracle $\mathcal{O}$ constructed on this basis will therefore
return a unique solution, and supports the theoretical
characterization of \emph{RAD} above.

The results above are encouraging, but many learning problems do not
feature a convex $f(x)$, or a continuous feature space.
Next, we consider several general algorithms for adversarial evasion.

\subsection{Coordinate Greedy}
\label{s:coor}

We propose a very general local search framework,
\emph{CoordinateGreedy (CG)} (Algorithm~\ref{alg:coordinate} for approximating
optimal attacker evasion.
\begin{algorithm}[thb]
\caption{CoordinateGreedy(\emph{CG}): $\mathcal{O}(\beta,x)$}
\label{alg:coordinate}
\begin{algorithmic}[1]
\STATE \textbf{Input}: {Parameter vector $\beta$, malicious instance $x$}
 \STATE {Set $k \leftarrow 0$ and let $x^0 \leftarrow x$}
 \REPEAT 
 \STATE{Randomly choose index $i_k \in \{1, 2, ..., n\}$ }
 \STATE{$x^{k+1} \leftarrow $GreedyImprove($i_k$)}
 \STATE {$k \leftarrow k+1$}
 \UNTIL{$\frac{\ln Q(x^{k})} {\ln Q(x^{k-1})} \leq \epsilon$}
\IF{$f(x^k) \ge 0$}
 \STATE $x^k \leftarrow x$
\ENDIF
 \STATE \textbf{Output}: {Adversarially optimal instance $x^k$.}
\end{algorithmic}
\end{algorithm}
The high-level idea is to iteratively choose a feature, and greedily
update this feature to incrementally improve the attacker's utility
(as defined by Problem~\ref{eq:attacker_q}).
In general, this algorithm will only converge to a locally optimal
solution.
We therefore propose a version with random restarts: run \emph{CG}
from $L$ random starting points in feature space.
As long as a global optimum has a basin of attraction with positive
Lebesgue measure, or the feature space is finite, this process will asymptotically converge to a
globally optimal solution as we increase the number of random
restarts.
Thus, as we increase the number of random restarts, we expect to increase
the frequency that we actual return the global optimum.
Let $p_L$ denote the probability that the oracle based on coordinate
greedy with $L$ random restarts returns a suboptimal solution to Problem~\ref{eq:attacker_q}.
The next result generalizes the bound on \emph{RAD} to allow for this,
restricting however that the risk function which we bound from above
uses the $0/1$ loss.
Let $\mathcal{L}_{A,01}^*(\mathcal{O})$ correspond to the total
adversarial risk in Equation~\ref{E:advloss}, where the loss function
$l(g_\beta(x),y)$ is the $0/1$ loss.
Suppose that $\mathcal{O}_L$ uses coordinate greedy with $L$ random restarts.
\begin{proposition}
\label{prop:defender_loss_gen}
Let $B = |I_{bad}|$.
$\mathcal{L}_{A,01}^*(\mathcal{O}) \le \mathcal{L}^R_N(\beta,
\mathcal{O}_L) + \delta(p)$ with probability at least $1-p$, where 
$\delta(p) = B \left(p_L + \frac{\sqrt{\log^2 p-8 B p_l \log p}-\log p}{2B}\right),$
and $\mathcal{L}^R_N(\beta,
\mathcal{O}_L)$ uses any loss function $l(g_\beta(x),y)$ which is an upper bound on
the $0/1$ loss.
\end{proposition}
Experiments suggest that $p_L \rightarrow 0$ quite rapidly for an array of learning algorithms, and for either
discrete or continuous features, as we increase the number of restarts
$L$ (see the supplement for details).
Consequently, in practice retraining with coordinate
greedy nearly minimizes an upper bound on minimal adversarial
risk for a 0/1 loss with few restarts of the approximate attacker oracle.

\noindent{\bf Continuous Feature Space: }
For continuous feature space, we assume that both $f(x)$ and
$c(x,\cdot)$ are differentiable in $x$, and propose using the \emph{coordinate descent}
algorithm, which is a special case of coordinate greedy, where the
GreedyImprove step is:
\(
x^{k+1} \leftarrow x^k - \tau_k e_{i_k} \frac{\partial
  Q(x^k)}{\partial x^k_{i_k}},
\)
where $\tau_k$ is the step size and $e_{i_k}$ the direction of $i_k$th
coordinate.
Henceforth, let the origial adversarial instance $x_i$ be given; we
then simplify cost function to be only a function of $x$, denoted $c(x)$.
If the function $f(x)$ is convex and differentiable, our coordinate
descent based algorithm~\ref{alg:coordinate} can always find the
global optima which is the attacker best response $x^*$
\cite{luo1992convergence}, and Proposition~\ref{prop:defender_loss}
applies, by Corollary~\ref{C:convex}.
If $f(x)$ is not convex, then coordinate descent will only converge
to a local optimum.

\noindent{\bf Discrete Feature Space: }
In the case of discrete feature space, GreedyImprove step of \emph{CG}
can simply
enumerate all options for feature $j$, and choose the one most
improving the objective.

\leaveout{
\subsection{Attacks as Constrained Optimization}

A variation on the attack models above is when the attacker is solving
the following constrained optimization problem:
\begin{subequations}
\label{E:constrainedAttacks}
\begin{align}
&\min_x \min\{0,f(x)\}\\
&\mathrm{s.t.:} \quad c(x,x_i) \le B
\end{align}
\end{subequations}
for some cost budget constraint $B$.
While this problem is, again, non-convex, we can instead minimize the
convex upper bound, $f(x)$, as before, if we assume that $f(x)$ is
convex.
In this case, if the feature space is continuous, the problem can be
solved optimally using standard convex optimization
methods~\cite{Boyd04}.
If the feature space is binary and $f(x)$ is linear or
convex-inducing, algorithms proposed by Lowd and
Meek~\cite{lowd2005adversarial} and Nelson et
al.~\cite{nelson2012query}.
}

\section{Experimental Results}

The results above suggest that the proposed systematic retraining
algorithm is likely to be effective at increasing resilience to
adversarial evasion.
We now offer an experimental evaluation of this (additional results
are provided in the supplement).
We present the results for the exponential cost model, where $c(x,x_i) = \exp \left(\lambda (\sum_j (x_j -x_{ij})^2 +1)^{1/2}
\right)$.
Additionally, we simulated attacks using Problem~\ref{eq:attacker_q} formulation.
Results for other cost functions and attack models are similar, as
shown in the supplement.
Moreover, the supplement demonstrates that the approach is robust to
cost function misspecification.



\noindent{\bf Comparison to Optimal: }
The first comparison we draw is to a recent algorithm, \emph{SMA},
which minimizes $l_1$-regularized adversarial risk function~\eqref{E:advloss} using
the hinge loss function.
Specifically, \emph{SMA} formulates the problem as a large mixed-integer
linear program which it solves using constraint
generation~\cite{li2014feature}.
The main limitation of SMA is scalability.
Because retraining methods use out-of-the-box learning tools and does
not involve non-convex bi-level optimization, it is
considerably more scalable.
\begin{figure}[H]
\centering
\begin{tabular}{cc}
\includegraphics[width=1.45in,height=1.19in]{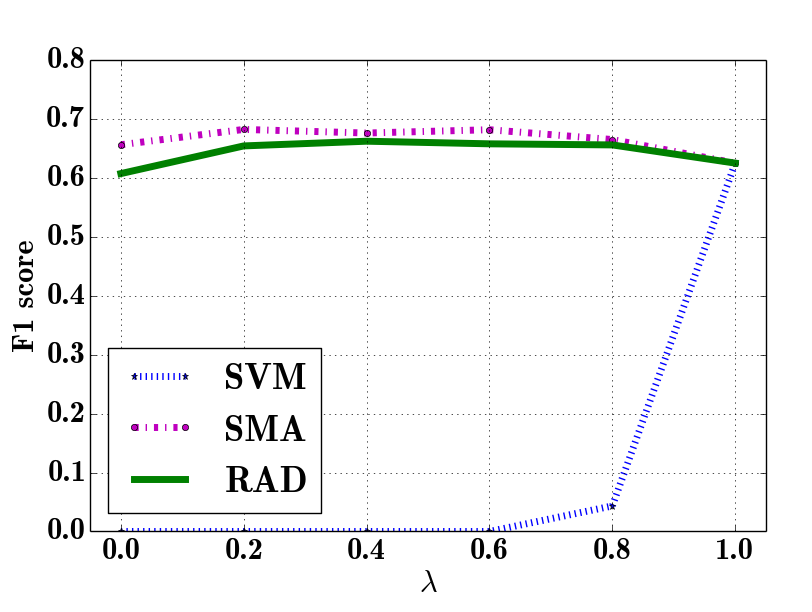} &
\includegraphics[scale=0.193]{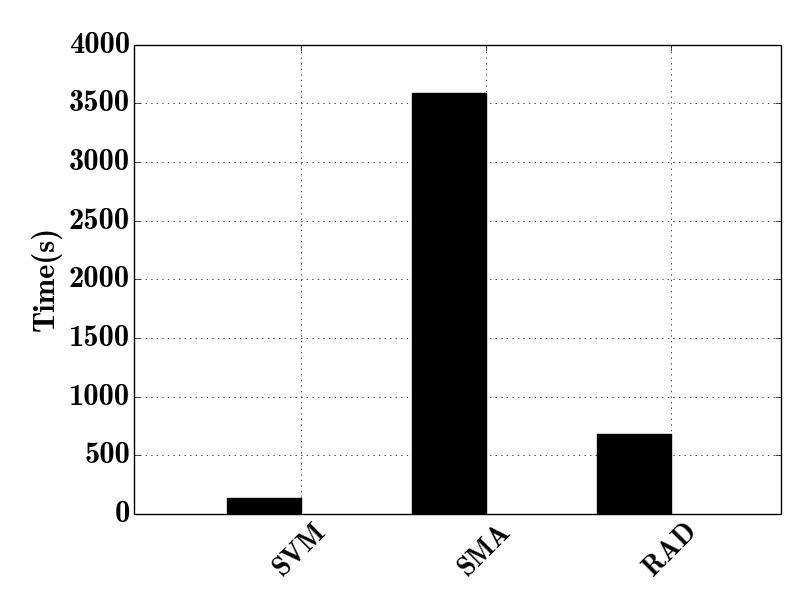} \\
(a) & (b)
\end{tabular}
\caption{Comparison between \emph{RAD} and \emph{SMA} based on the Enron dataset with 30 binary features. (a) The $F_1$ score of different algorithms corresponding to various $\lambda$; (b) the average runtime for each algorithm.}
\label{f:sma_compare}
\end{figure} 

We compared \emph{SMA} and \emph{RAD} using
Enron data~\cite{klimt2004enron}. 
As Figure~\ref{f:sma_compare}(a) demonstrates, retraining solutions of \emph{RAD} are
nearly as good as \emph{SMA}, particularly for a non-trivial adversarial cost sensitive $\lambda$.
In contrast, a baseline implementation of SVM is significantly more
fragile to evasion attacks. However, the runtime comparison for these
algorithms in Figure~\ref{f:sma_compare}(b) shows that \emph{RAD} is
much more scalable than \emph{SMA}.
\leaveout{
In addition, Figure~\ref{f:sma_compare}(a) also reveals a surprising phenomenon observed by
Kantchelian et al.~\cite{Kantchelian15} as well: performance
actually improves based on certain adversarial cost sensitivity compared with adversary-free case where high $\lambda$ value is applied.
What this shows is that adding certain adversarial instances actually improves
classification robustness to non-adversarial data as well, likely
because it makes the learner significantly more robust to noise in the
data. 
}


\noindent{\bf Effectiveness of Retraining: }
In this section we use the Enron dataset~\cite{klimt2004enron} and
MNIST~\cite{lecun2010mnist} dataset to evaluate the robustness of
three common algorithms in their standard implementation, and in
\emph{RAD}: logistic
regression, SVM (using a linear kernel), and a neural network (NN)
with 3 hidden layers.
In Enron data, features correspond to relative word frequencies.
2000 features were
used for the Enron and 784 for MNIST datasets.
Throughout, we use precision, recall, and accuracy as metrics.
We present the results for a continuous feature space here.
Results for binary features are similar and provided in the supplement.

\begin{figure*}[t]
\centering 
\begin{tabular}{cccc}
\includegraphics[scale=0.193]{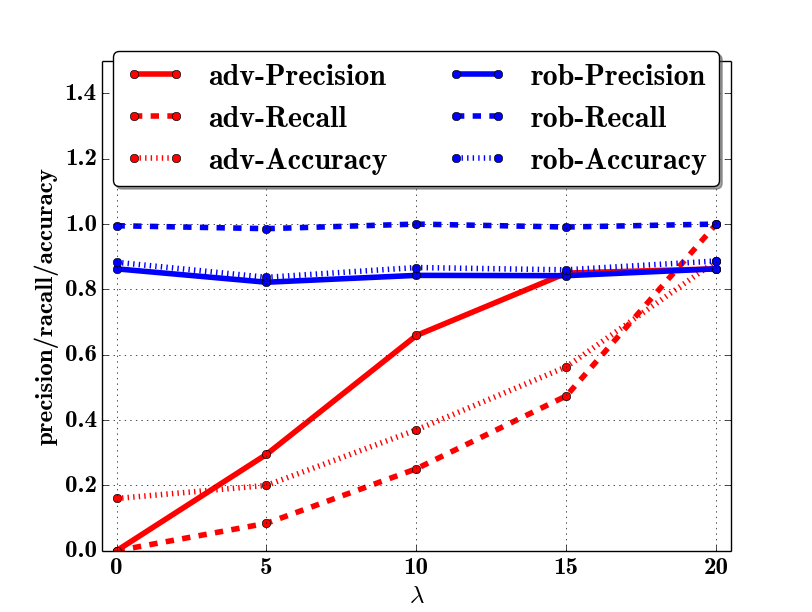} &
\includegraphics[scale=0.193]{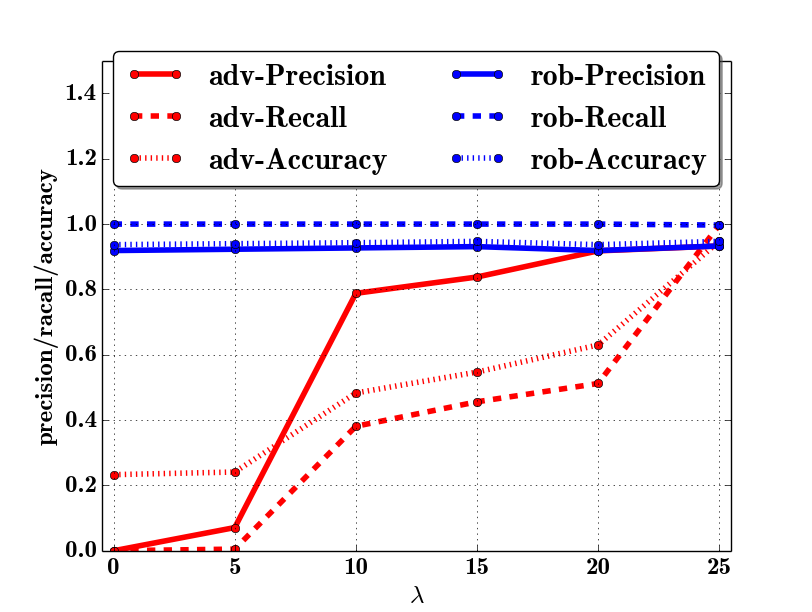} &
\includegraphics[scale=0.193]{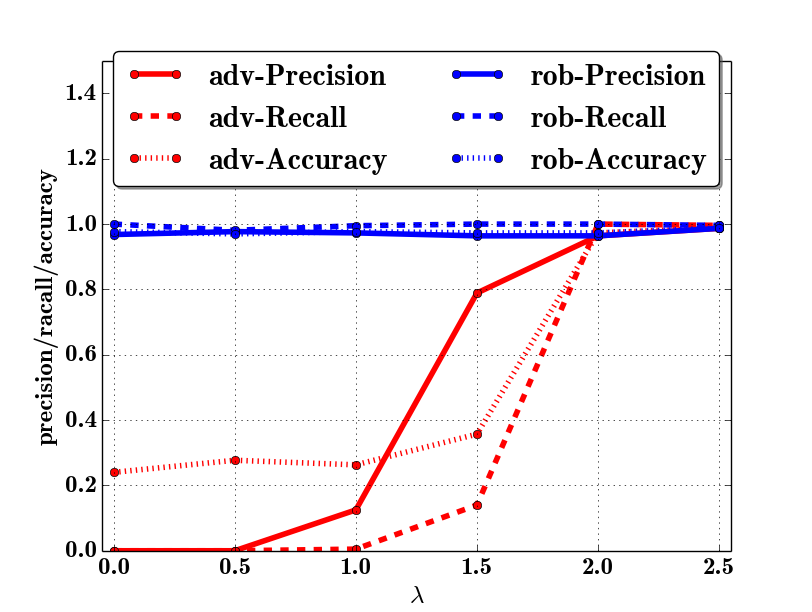} \\
\includegraphics[scale=0.193]{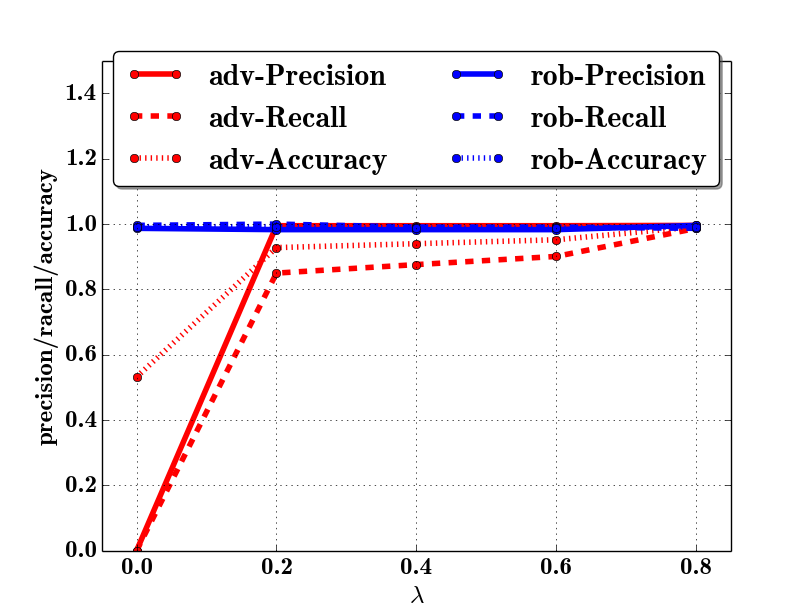} &
\includegraphics[scale=0.193]{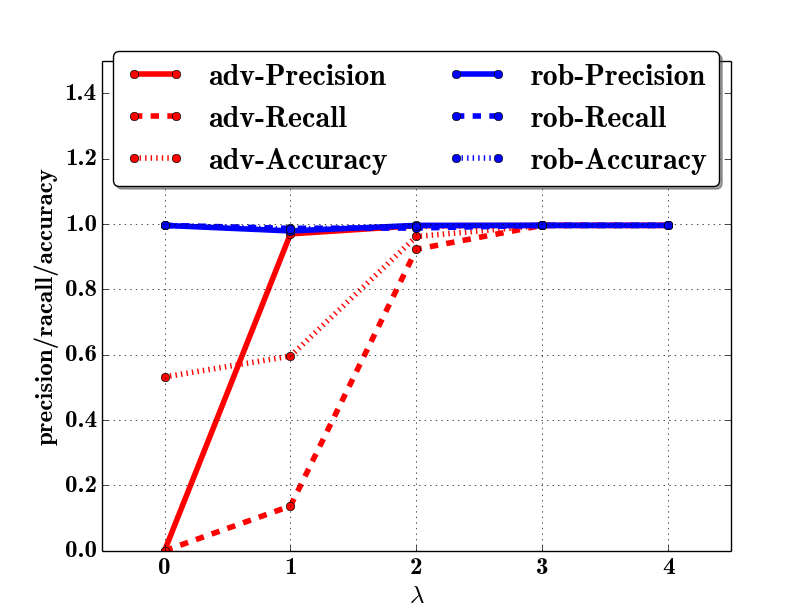} &
\includegraphics[scale=0.193]{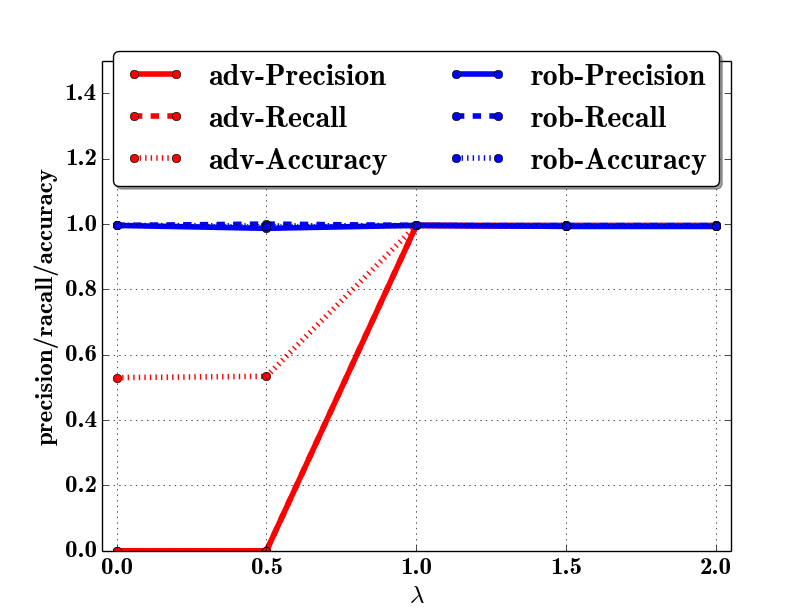}\\
(a) & (b) & (c) 
\end{tabular}
\caption{Performance of baseline (\emph{adv-}) and \emph{RAD} (\emph{rob-}) as a function of cost
  sensitivity $\lambda$ for Enron (top) and MNIST (bottom) datasets with continuous features testing on adversarial instances. (a)
  logistic regression, (b) SVM, (c) 3-layer NN.}
\label{fig:log_con}
\end{figure*}

\begin{figure*}[t]
\centering 
\begin{tabular}{ccc}
\includegraphics[scale=0.193]{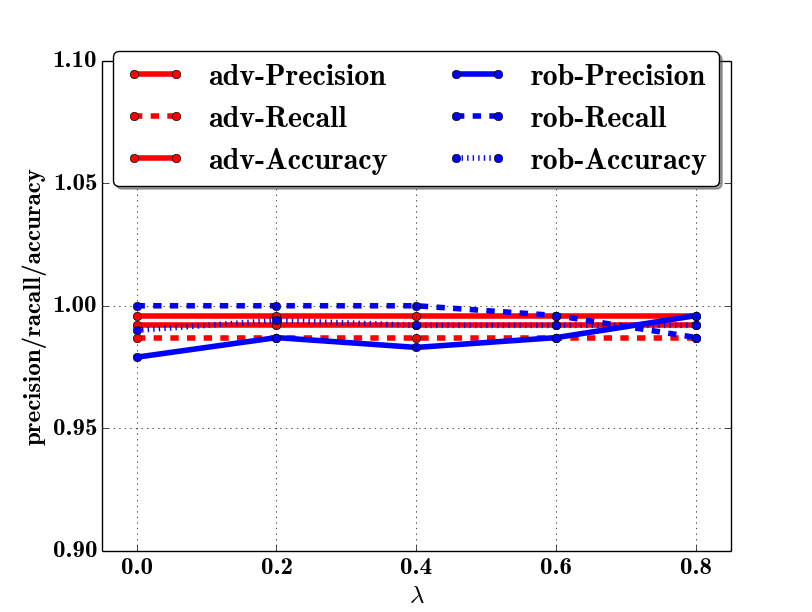} &
\includegraphics[scale=0.193]{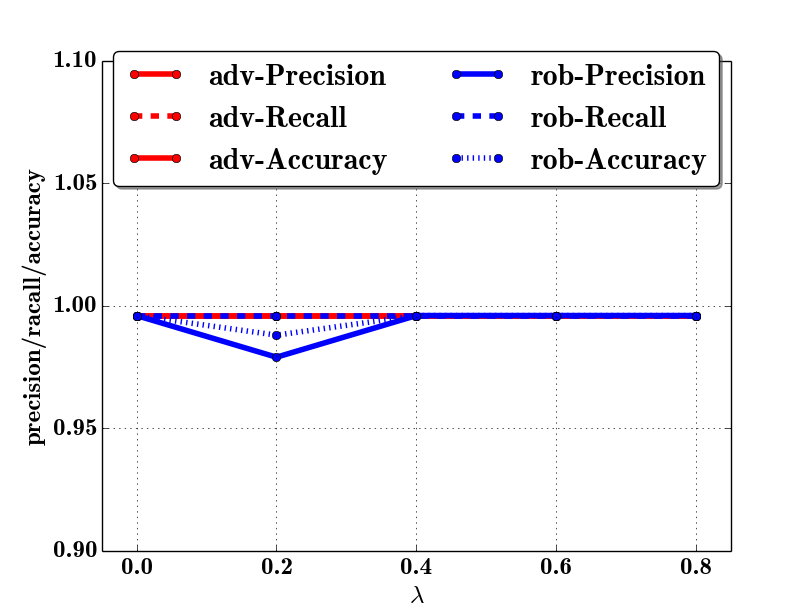} &
\includegraphics[scale=0.193]{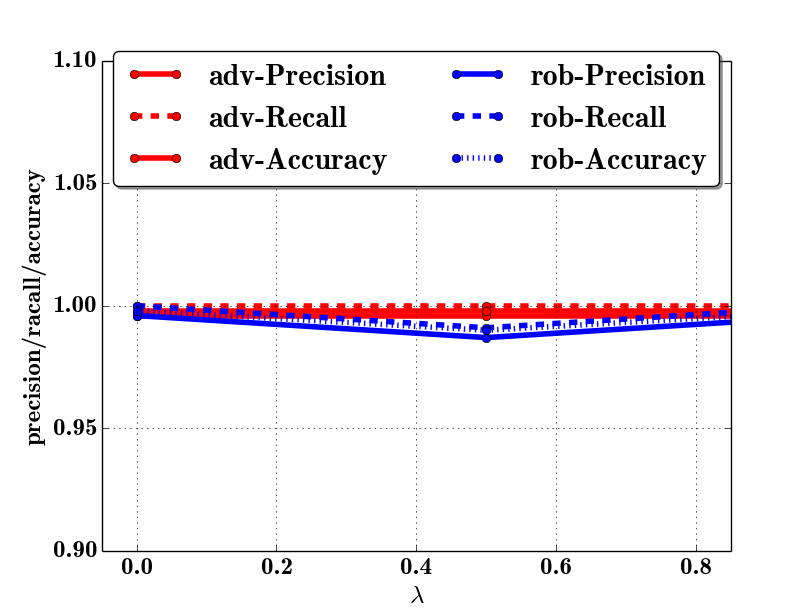} \\
(a) & (b) & (c)
\end{tabular}
\caption{Performance of baseline (\emph{adv-}) and \emph{RAD} (\emph{rob-}) as a function of cost
sensitivity $\lambda$ for MNIST dataset with continuous features testing on non-adversarial instances. (a) logistic regression, (b) SVM, (c) 3-layer NN.}
\label{fig:cleantest_con}
\end{figure*}

Figure~\ref{fig:log_con}(a) shows the performance of logistic regression,
with and without retraining, on Enron and MNIST.
The increased robustness of \emph{RAD} is immediately evident:
performance of \emph{RAD} is essentially independent of $\lambda$ on
all three measures, and substantially exceeds baseline algorithm
performance for small $\lambda$.
Interestingly, we observe that the baseline algorithms
are significantly more fragile to evasion attacks on Enron data
compard to MNIST: benign and malicious classes seem far easier to
separate on the latter than the former.
This qualitative comparison between the Enron and MNIST datasets is
consistent for other classification methods as well (SVM, NN).
These results also illustrate that the neural-network classifiers, in their baseline implementation, are
significantly more robust to evasion attacks than the (generalized) linear
classifiers (logistic regression and SVM): even with a relatively
small attack cost attacks become ineffective relatively quickly, and
the differences between the performance on Enron and MNIST data are
far smaller.
Throughout, however, \emph{RAD} significantly improves robustness to
evasion, maintaining extremely high accuracy, precision, and recall
essentially independently of $\lambda$, dataset, and algorithm used.

In order to explore whether \emph{RAD} would sacrifice accuracy
when no adversary is present, Figure~\ref{fig:cleantest_con} shows the performance of the
baseline algorithms and \emph{RAD} on a test dataset sans
evasions. Surprisingly, \emph{RAD} is never significantly worse, and
in some cases better than non-adversarial baselines: adding malicious instances appears to increase overall generalization ability.
This is also consistent with the observation by Kantchelian et al.~\cite{Kantchelian15}.

\noindent{\bf Oracles based on Human Evasion Behavior: }
To evaluate the considerable generality of \emph{RAD}, we now use a
non-optimization-based threat model, making use instead
of observed human evasion behavior \emph{in human subject experiments}.
The data for this evaluation was obtained from the human subject
experiment by Ke et al.~\cite{ke2016behavioral} in which subjects were
tasked with the goal of evading an SVM-based spam filter, manipulating
10 spam/phishing email instances in the process.
In these experiments, Ke et al. used machine learning to develop a
model of human subject evasion behavior.
We now adopt this model as the evasion oracle, $\mathcal{O}$, injected
in our \emph{RAD} retraining framework, executing the synthetic model
for 0-10 iterations to obtain evasion examples.

Figure~\ref{fig:behavior}(a) shows the recall results for the dataset of
10 malicious emails (the classifiers are trained on Enron data, but
evaluated on these 10 emails, including evasion
attacks). Figure~\ref{fig:behavior}(b) shows the classifier performance for the Enron dataset by applying the synthetic adversarial model as the oracle.
\begin{figure}[h]
\centering 
\begin{tabular}{cc}
\includegraphics[scale=0.193]{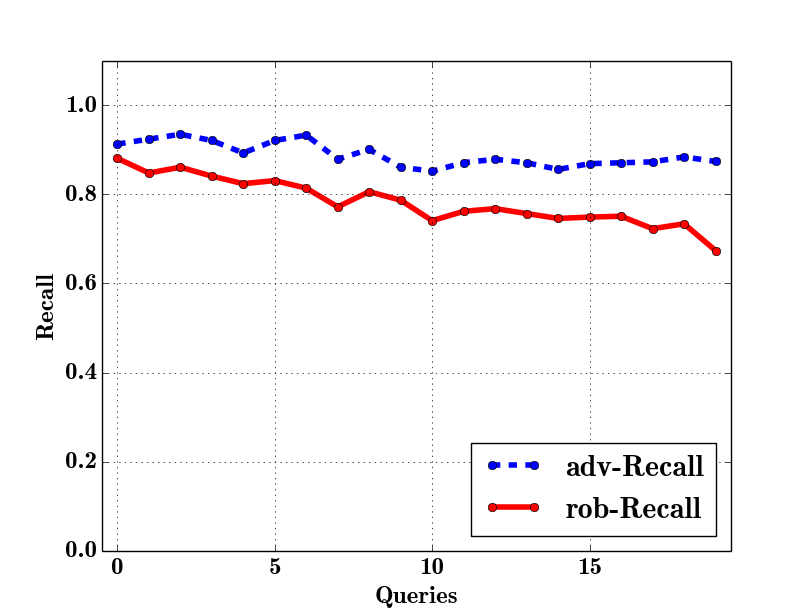} &
\includegraphics[scale=0.193]{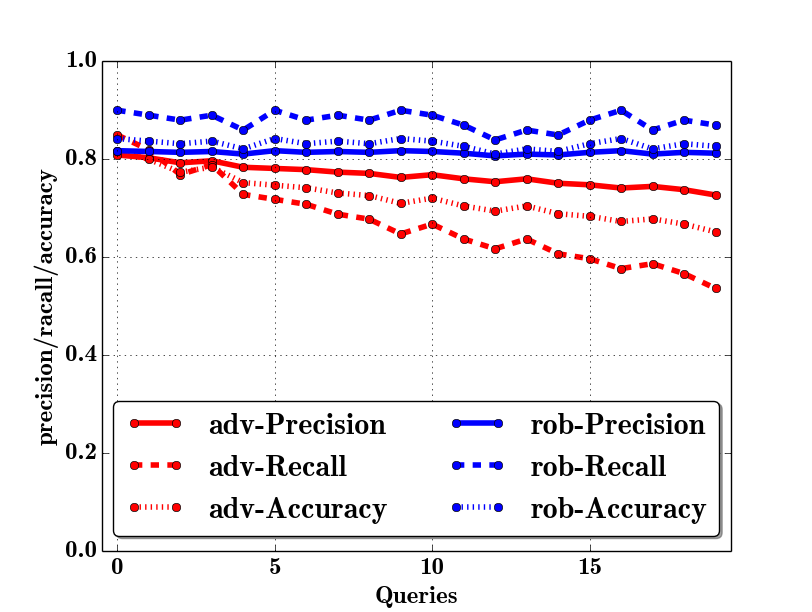}\\
(a) & (b)
\end{tabular}
\caption{\emph{RAD} (\emph{rob-}) and baseline SVM (\emph{adv-}) performance based on human subject behavior
  data over 20 queries, (a) using experimental data with actual human
  subject experiment submissions, (b) using Enron data and a synthetic
  model of human evader.}
\label{fig:behavior}
\end{figure}
We can make two high-level observations.
First, notice that human adversaries appear significantly less
powerful in evading the classifier than the automated
optimization-based attacks we previously considered.
This is a testament to both the effectiveness of our general-purpose adversarial
evaluation approach, and the likelihood that such automated attacks
likely significantly overestimate adversarial evasion risk in many
settings.
Nevertheless, we can observe that the synthetic model used in
\emph{RAD} leads to a significantly more robust classifier.
Moreover, as our evaluation used actual evasions, while the synthetic
model was used only in training the classifier as a part of
\emph{RAD}, this experiment suggests that the synthetic model can be
relatively effective in modeling behavior of human adversaries.
Figure~\ref{fig:behavior}(b) performs a more systematic study
using the synthetic model of adversarial behavior on the Enron dataset.
The findings are consistent with those only considering the 10 spam
instances: retraining significantly boosts robustness to evasion, with
classifier effectiveness essentially independent of the number of
queries made by the oracle.
\section{Conclusion}

We proposed a general-purpose systematic retraining algorithm against evasion attacks of classifiers for arbitrary oracle-based evasion models.
We first demonstrated that this algorithm effectively minimizes an upper bound on optimal adversarial risk, which is typically extremely difficult to compute (indeed, no approach exists for minimizing adversarial loss for an arbitrary evasion oracle).
Experimentally, we showed that the performance of our retraining approach is nearly indistinguishable from optimal, whereas scalability is dramatically improved: indeed, with \emph{RAD}, we are able to easily scale the approach to thousands of features, whereas a state-of-the-art adversarial risk optimization method can only scale to 15-30 features.
We generalize our results to show that a probabilistic upper bound on minimal adversarial loss can be obtained even when the oracle is computed approximately by leveraging random restarts, and an empirical evaluation which confirms that the resulting bound relaxation is tight in practice.


We also offer a general-purpose framework for optimization-based oracles using variations of coordinate greedy algorithm on both discrete and continuous feature spaces.
Our experiments demonstrate that our adversarial oracle approach is extremely effective in corrupting the baseline learning algorithms.
On the other hand, extensive experiments also show that the use of our retraining methods significantly boosts robustness of algorithms to evasion.
Indeed, retrained algorithms become nearly insensitive to adversarial evasion attacks, at the same time maintaining extremely good learning performance on data overall.
Perhaps the most significant strength of the proposed approach is that it can make use of arbitrary learning algorithms essentially ``out-of-the-box'', and effectively and quickly boost their robustness, in contrast to most prior adversarial learning methods which were algorithm-specific.




\small
\bibliography{nips_2016_retraining}
\bibliographystyle{unsrt}

\newpage
\normalsize
\section*{Supplement To: A General Retraining Framework for Scalable Adversarial Classification}
\subsection*{Proof of Lemma 4.1}

If $f(x^*) < 0$, then $\min\{0,f(x^*)\} + c(x^*,x_i) = f(x^*) +
c(x^*,x_i)$.
By optimality of $\bar{x}$, $f(\bar{x}) + c(\bar{x},x_i) \le f(x^*) +
c(x^*,x_i)$.
Since $x_i$ is suboptimal in Problem~(4) and $c$ strictly
positive in all other cases, $f(x^*) +
c(x^*,x_i) < \min\{0,f(x_i)\} + c(x_i,x_i) =0$.
By optimality of $x^*$, $f(x^*) + c(x^*,x_i) \le \min\{0,f(\bar{x})\}
+ c(\bar{x},x_i) \le f(\bar{x}) + c(\bar{x}, x_i)$, which implies that
$f(\bar{x}) + c(\bar{x},x_i) = f(x^*) + c(x^*,x_i)$.
Consequently, $f(\bar{x}) + c(\bar{x},x_i) < 0$, and, therefore,
$f(\bar{x}) < 0$.

\subsection*{Proof of Proposition 4.2}

Let $\bar{\beta} \in \arg\min_{\beta}  \mathcal{L}^R_N(\beta,
\mathcal{O}_L)$.
Consequently, for any $\beta$,
\begin{align*}
\mathcal{L}_{A,01}^*(\mathcal{O}_L) &= \min_\beta \mathcal{L}_{A,01}(\beta;
\mathcal{O}_L) \\
&\le \sum_{i: y_i = -1} l_{01}(g_{\bar{\beta}}(x_i),-1) +\sum_{i:y_i = +1}
l_{01}(g_{\bar{\beta}} (\mathcal{O}(\bar{\beta},x_i)),+1) + \alpha ||\bar{\beta}||_p^p.
\end{align*}
Now,
\begin{align*}
\sum_{i:y_i = +1} &l_{01}(g_{\bar{\beta}}
(\mathcal{O}(\bar{\beta},x_i)),+1) \le \sum_{i:y_i = +1} l_{01}(g_{\bar{\beta}}
(\mathcal{O}_L(\bar{\beta},x_i)),+1) + \delta(p)
\end{align*}
with probability at least $1 - p$,
where $\delta(p) = B p_L + \frac{\sqrt{\log^2 p-8 B p_l \log p}-\log p}{2}$,
by the Chernoff bound, and Lemma~4.1, which assures that an
optimal solution to Problem~(5) can only over-estimate mistakes.
Moreover,
\begin{align*}
\sum_{i:y_i = +1} &l_{01}(g_{\bar{\beta}}
(\mathcal{O}_L(\bar{\beta},x_i)),+1) \le\sum_{i:y_i = +1}\sum_{j \in N_i} l(g_{\bar{\beta}}
(\mathcal{O}_L(\bar{\beta},x_i)),+1),
\end{align*}
since $\mathcal{O}_L(\bar{\beta},x_i) \in N_i$ for all $i$ by
construction, and $l$ is an upper bound on $l_{01}$.
Putting everything together, we get the desired result.

\subsection*{Convergence of $p_L$ with Increasing Number of Restarts $L$}

\begin{figure}[H]
\centering 
\begin{tabular}{cc}
\includegraphics[scale=0.23]{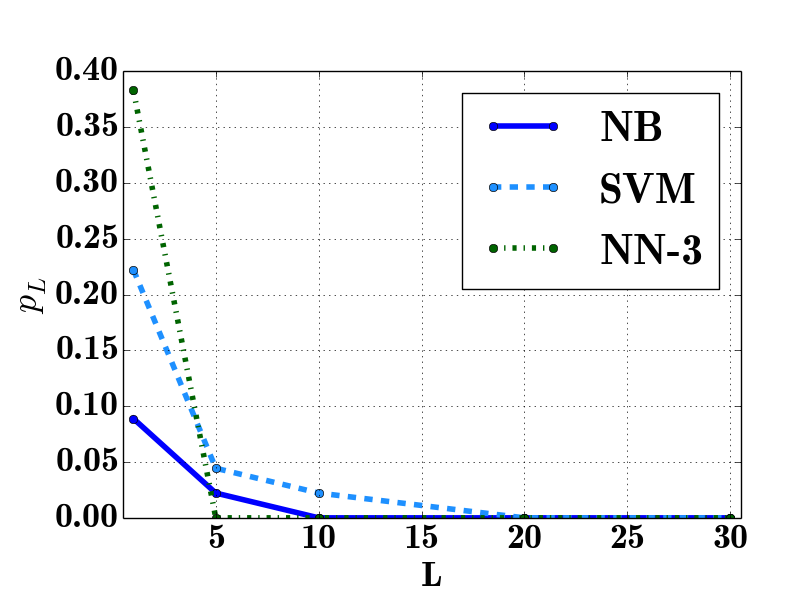} &
\includegraphics[scale=0.23]{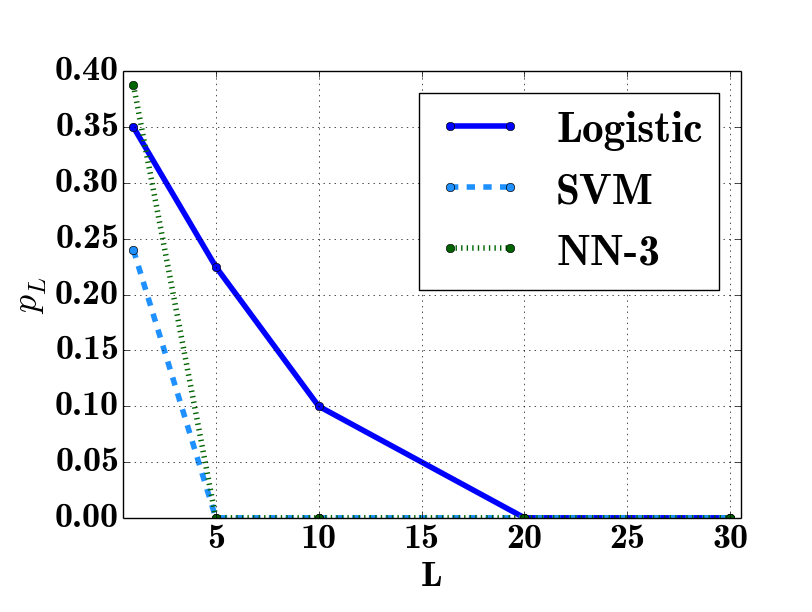} \\
(a) & (b)
\end{tabular}
\caption{The convergence of $p_L$ based on different number of starting points for (a) Binary, (b) Continuous feature space.}
\label{fig:start}
\end{figure}

\subsection*{Attacks as Constrained Optimization}

A variation on the attack models in the main paper is when the attacker is solving
the following constrained optimization problem:
\begin{subequations}
\label{E:constrainedAttacks}
\begin{align}
&\min_x \min\{0,f(x)\}\\
&\mathrm{s.t.:} \quad c(x,x_i) \le B
\end{align}
\end{subequations}
for some cost budget constraint $B$ and query budget constraint $Q$.
While this problem is, again, non-convex, we can instead minimize the
convex upper bound, $f(x)$, as before, if we assume that $f(x)$ is
convex.
In this case, if the feature space is continuous, the problem can be
solved optimally using standard convex optimization
methods~\cite{Boyd04}.
If the feature space is binary and $f(x)$ is linear or
convex-inducing, algorithms proposed by Lowd and
Meek~\cite{lowd2005adversarial} and Nelson et
al.~\cite{nelson2012query}. Figure~\ref{fig:constraint_nb}, \ref{fig:constraint_svm} and~\ref{fig:constraint_3nn} show the performance of \emph{RAD} based on the optimized adversarial strategies for various learning models, respectively. 

\begin{figure}[tbfh]
\centering 
\begin{tabular}{cc}
\includegraphics[scale=0.23]{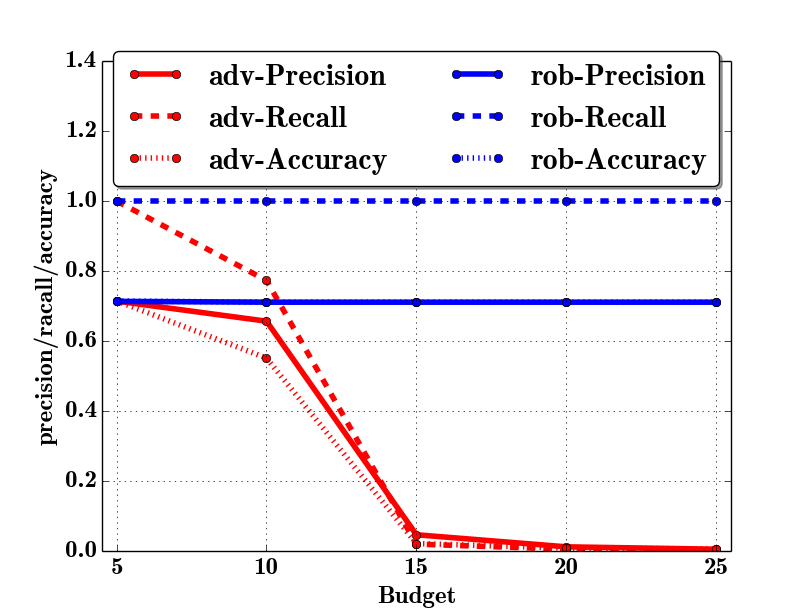} &
\includegraphics[scale=0.23]{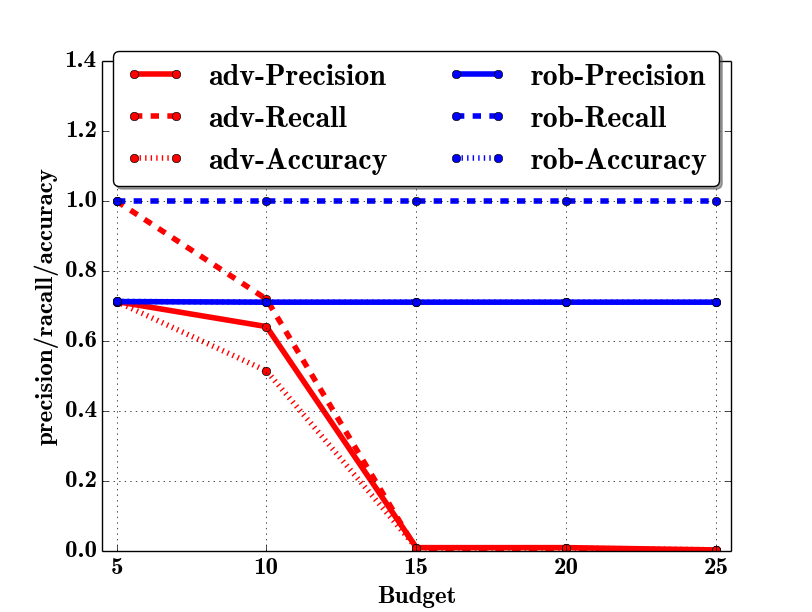} \\
(a) & (b)
\end{tabular}
\caption{Performance of baseline (\emph{adv-}) and \emph{RAD} (\emph{rob-}) as a function of adversarial budget for Enron dataset with binary features testing on adversarial instances using Naive Bayesian. (a) query budget $Q=30$, (b) query budget $Q=40$.}
\label{fig:constraint_nb}
\end{figure}

\begin{figure}[tbfh]
\centering 
\begin{tabular}{cc}
\includegraphics[scale=0.23]{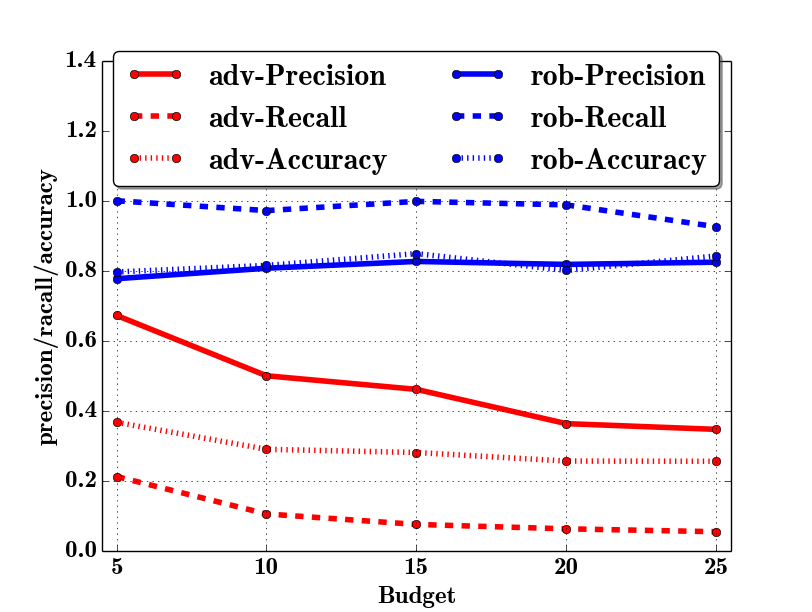} &
\includegraphics[scale=0.23]{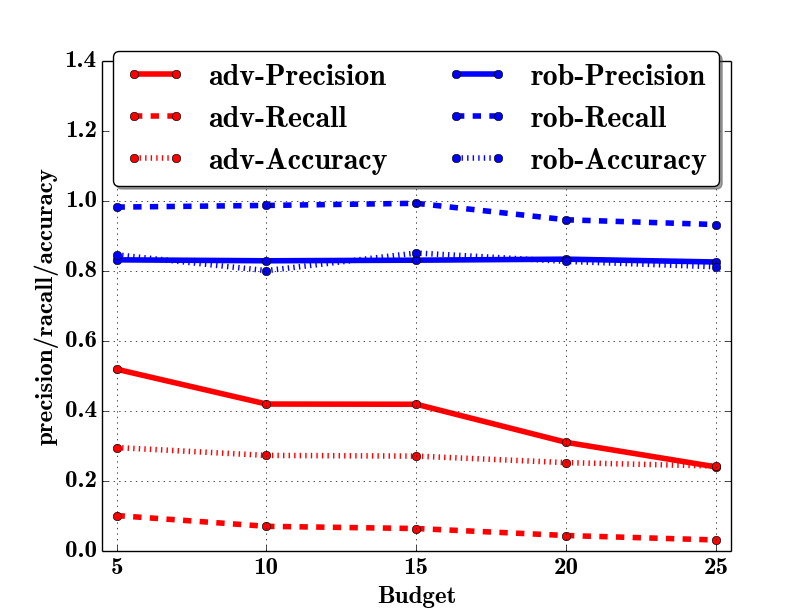} \\
(a) & (b)
\end{tabular}
\caption{Performance of baseline (\emph{adv-}) and \emph{RAD} (\emph{rob-}) as a function of adversarial budget for Enron dataset with binary features testing on adversarial instances using SVM with RBF kernel. (a) query budget $Q=30$, (b) query budget $Q=40$.}
\label{fig:constraint_svm}
\end{figure}

\begin{figure}[tbfh]
\centering 
\begin{tabular}{cc}
\includegraphics[scale=0.23]{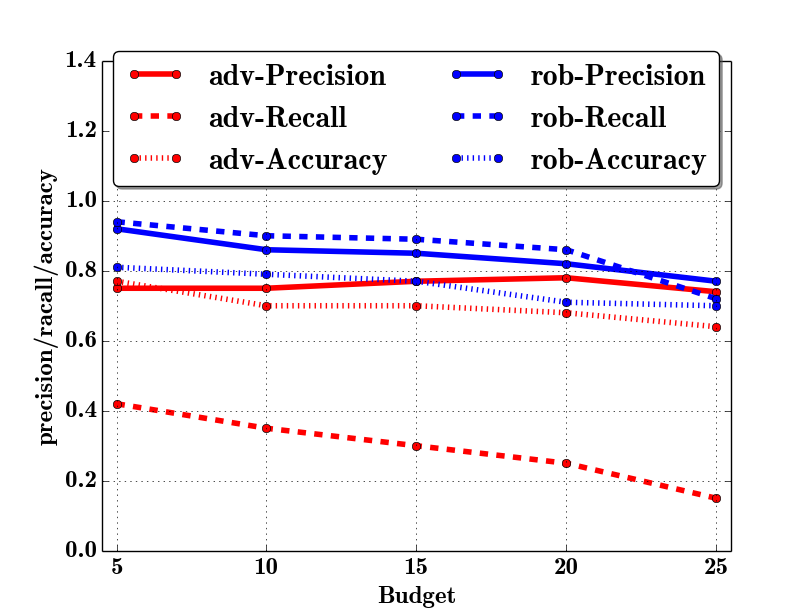} &
\includegraphics[scale=0.23]{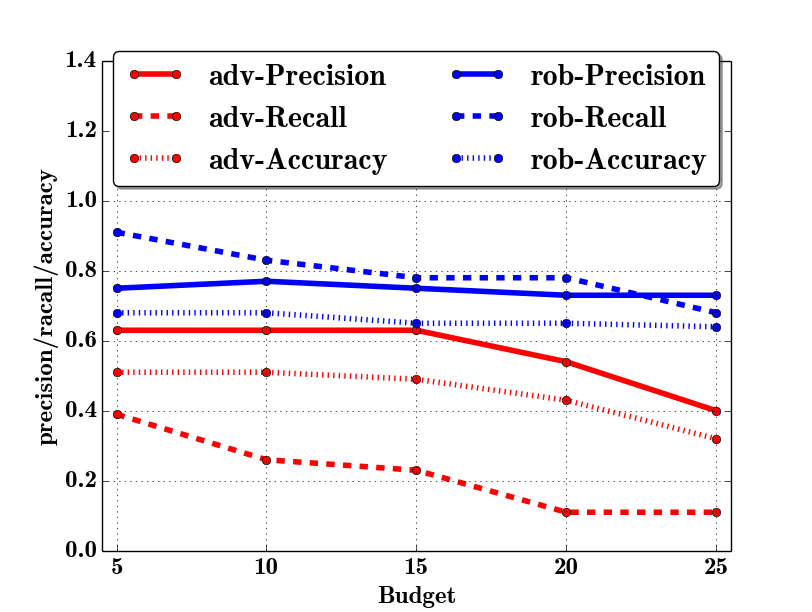} \\
(a) & (b)
\end{tabular}
\caption{Performance of baseline (\emph{adv-}) and \emph{RAD} (\emph{rob-}) as a function of adversarial budget for Enron dataset with binary features testing on adversarial instances using 3-layer NN. (a) query budget $Q=30$, (b) query budget $Q=40$.}
\label{fig:constraint_3nn}
\end{figure}

\subsection*{Experiments with Continuous Feature Space}

\begin{figure}[H]
\centering 
\begin{tabular}{cccc}
\includegraphics[scale=1.0]{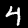} &
\includegraphics[scale=1.0]{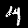} &
\includegraphics[scale=1.0]{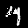} &
\includegraphics[scale=1.0]{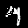}\\
\includegraphics[scale=1.0]{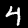} &
\includegraphics[scale=1.0]{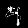} &
\includegraphics[scale=1.0]{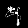} &
\includegraphics[scale=1.0]{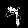} \\
\includegraphics[scale=1.0]{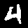} &
\includegraphics[scale=1.0]{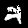} &
\includegraphics[scale=1.0]{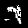} &
\includegraphics[scale=1.0]{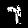}\\
\includegraphics[scale=1.0]{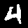} &
\includegraphics[scale=1.0]{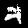} &
\includegraphics[scale=1.0]{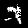} &
\includegraphics[scale=1.0]{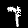}
\end{tabular}
\caption{Example modification of digit images (MNIST data) as
  $\lambda$ decreases (left-to-right) for logistic regression,
SVM, 1-layer NN, and 3-layer NN (rows 1-4 respectively).}
\label{fig:visual_logistic}
\end{figure}

In Figure~\ref{fig:visual_logistic} we visualize the relative
vulnerability of the different classifiers, as well as effectiveness
of our general-purpose evasion methods based on coordinate greedy.
Each row corresponds to a classifier, and moving right within a row
represents decreasing $\lambda$ (allowing attacks to make more
substantial modifications to the image in an effort to evade correct
classification).
We can observe that NN classifiers require more substantial changes to
the images to evade, ultimately making these entirely unlike the
original.
In contrast, logistic regression is quite vulnerable: the digit
remains largely recognizable even after evasion attacks.

\subsection*{Experiments with Discrete Feature Space}


Considering now data sets with binary features, we use the Enron data
with a bag-of-words feature representation, for a total of 2000
features.
We compare Naive Bayes (NB), logistic regression, SVM, and a 3-layer
neural network.
Our comparison involves both the baseline, and \emph{RAD}
implementations of these, using the same metrics as above.

\begin{figure*}[tbfh]
\centering 
\begin{tabular}{cccc}
\includegraphics[scale=0.21]{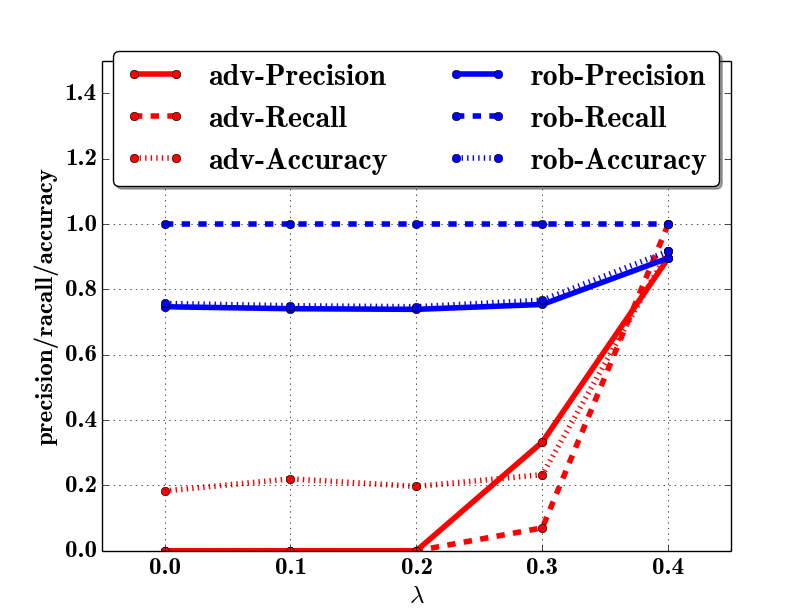} &
\includegraphics[scale=0.21]{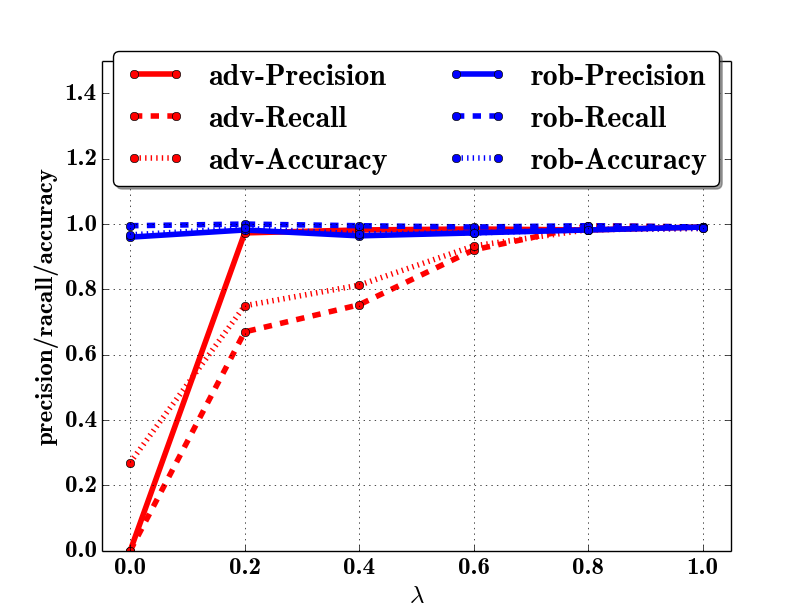} \\
(a) & (b) \\
\includegraphics[scale=0.21]{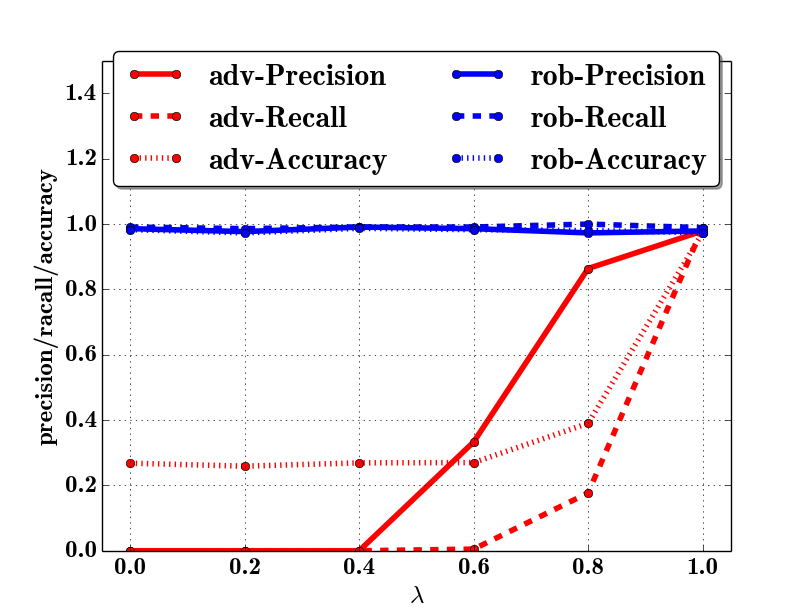} &
\includegraphics[scale=0.21]{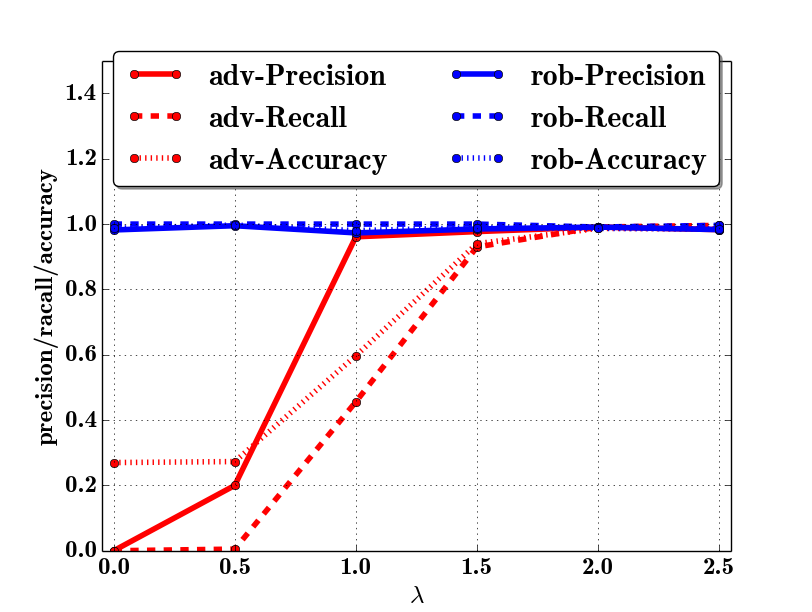}\\
(c) & (d)
\end{tabular}
\caption{Performance of baseline (\emph{adv-}) and \emph{RAD} (\emph{rob-}) implementations of (a) Naive
  Bayes, (b) logistic regression, (c) SVM, and (d) 3-layer NN, using binary features testing on adversarial instances.}
\label{fig:enron_binary}
\end{figure*}

Figure~\ref{fig:enron_binary} confirms the effectiveness of
\emph{RAD}: every algorithm is substantially more robust to evasion
with retraining, compared to baseline implementation.
Most of the algorithms can obtain extremely high accuracy on this data
with the bag-of-words feature representation.
However, a 3-layer neural network is now
\emph{less} robust than the other algorithms, unlike in the
experiments with continuous features.
Indeed, Goodfellow et
al.~\cite{goodfellow2014explaining} similarly observe the relative
fragility of NN to evasion attacks.

\subsection*{Experiments with Multi-class Classification}
\label{s:multiclass}

Discussion so far dealt entirely with binary classification.
We now observe that extending it to multi-class problems is quite direct.
Specifically, while previously the attacker aimed to make an instance classified as $+1$ (malicious) into a benign instance ($-1$), for a general label set $Y$, we can define a malicious set $M \subset Y$ and a target set $T \subset Y$, with $M \cap T = \emptyset$, where every entity represented by a feature vector $x$ with a label $y \in M$ aims to transform $x$ so that its label is changed to $T$.
In this setting, let $g(x) = \arg\max_{y \in Y} f(x,y)$.
We can then use the following empirical risk function:
\begin{equation}
\sum_{i : y_i \notin M} l( g_\beta(x_i),y_i) + \sum_{i:y_i \in M} l( g_\beta(\mathcal{O}(\beta,x_i)),y_i) + \lambda ||\beta||_p,
\end{equation}
where $\mathcal{O}$ aims to transform instances $x_i$ so that $g_\beta(\mathcal{O}(\beta,x_i)) \in T$.
The relaxed version of the adversarial problem can then be generalized to
\[
\min_{x,y \in T} -f(x,y) + c(x,x_i).
\]
For a finite target set $T$, this problem is equivalent to taking the best solution of a finite collection of problems identical to Problem 5.

To evaluate the effectiveness of \emph{RAD}, and resilience of
baseline algorithms, in multi-class classification settings, we use
the MNIST dataset and aim to correctly identify digits based on their
images.
Our comparison involves SVM and 3-layer neural network (results for
NN-1 are similar).
We use $M = \{1,4\}$ as the malicious class (that is, instances
corresponding to digits $1$ and $4$ are malicious), and $T = \{2,7\}$
is the set of benign labels (what malicious instances wish to be
classified as).
\begin{figure}[tbfh]
\centering 
\begin{tabular}{cccc}
\includegraphics[scale=0.23]{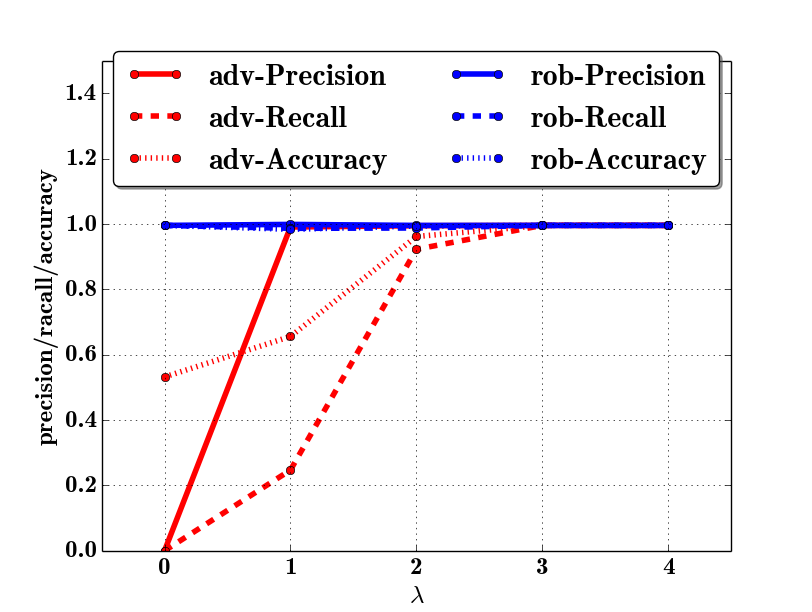} &
\includegraphics[scale=0.23]{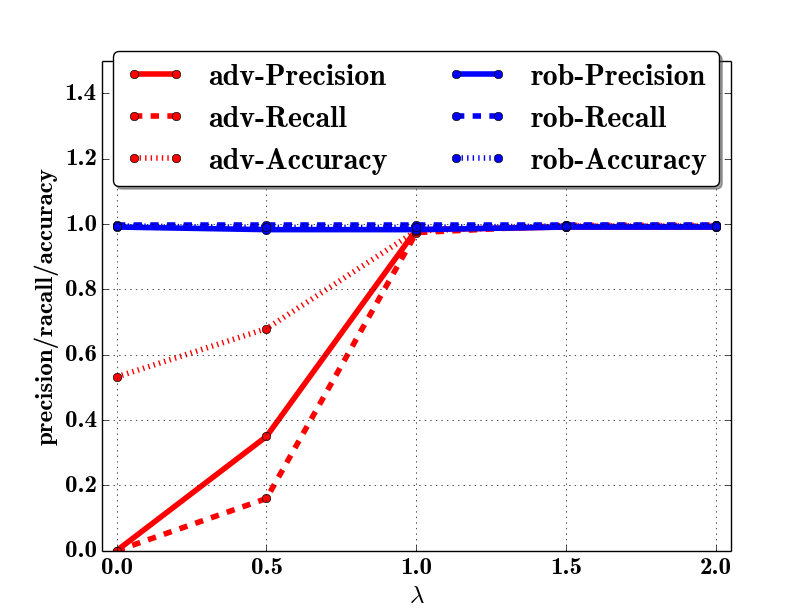} &
\end{tabular}
\caption{Performance of baseline (\emph{adv-}) and RAD (\emph{rob-}) implementations of (a) multi-class SVM
  and (b) multi-class 3-layer NN, using MNIST dataset testing on adversarial instances.}
\label{fig:enron_binary_multiclass}
\end{figure}
The results, shown in Figure~\ref{fig:enron_binary_multiclass} are largely
consistent with our previous observations: both SVM and 3-layer NN
perform well when retrained with \emph{RAD}, with near-perfect
accuracy despite adversarial evasion attempts.
Moreover, \emph{RAD} significantly boosts robustness to evasion,
particularly when $\lambda$ is small (adversary who is not very
sensitive to evasion costs).

\begin{figure}[h]
\centering
\begin{tabular}{cccc}
\includegraphics[scale=1.0]{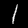} &
\includegraphics[scale=1.0]{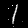} &
\includegraphics[scale=1.0]{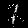} &
\includegraphics[scale=1.0]{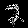} \\
\includegraphics[scale=1.0]{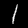} &
\includegraphics[scale=1.0]{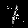} &
\includegraphics[scale=1.0]{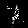} &
\includegraphics[scale=1.0]{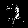} \\
\end{tabular}
\caption{Visualization of modification attacks with decreasing the
  cost sensitivity parameter $\lambda$ (from left to right), to change
  1 to the set \{2,7\}. The rows correspond to SVM and
  3-layer NN, respectively.}
\label{fig:visual_multi_1}
\end{figure}
Figure~\ref{fig:visual_multi_1} offers a visual demonstration of the relative
effectiveness of attacks on the baseline implementation of SVM and 1-
and 3-layer neural networks.
Here, we can observe that a significant change is required to evade
the linear SVM, with the digit having to nearly resemble a 2 after
modification.
In contrast, significantly less noise is added to the neural network
in effecting evasion.

\subsection*{Evaluation of Cost Function Variations and Robustness to
  Misspecification}
Considering the variations of cost functions, here we evaluate the classification efficiency for various cost functions as well as different cost functions for defender and adversary, respectively. 
Figure~\ref{fig:cost_variation} shows the empirical evaluation results based on Enron dataset with binary features. It is shown that if both the defender and adversary apply the L1 (a) or quadratic cost functions (b), it is easy to defend the malicious manipulations. Even the defender mistakenly evaluate the adversarial cost models as shown in Figure~\ref{fig:cost_variation} (c), $RAD$ framework can still defend the attack strategies efficiently.

\begin{figure*}[tbfh]
\centering 
\begin{tabular}{cccc}
\includegraphics[scale=0.193]{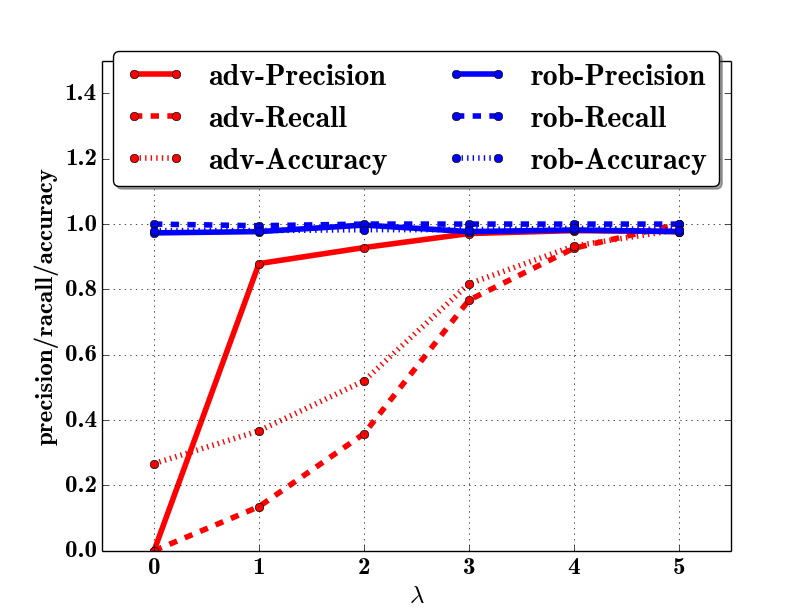} &
\includegraphics[scale=0.193]{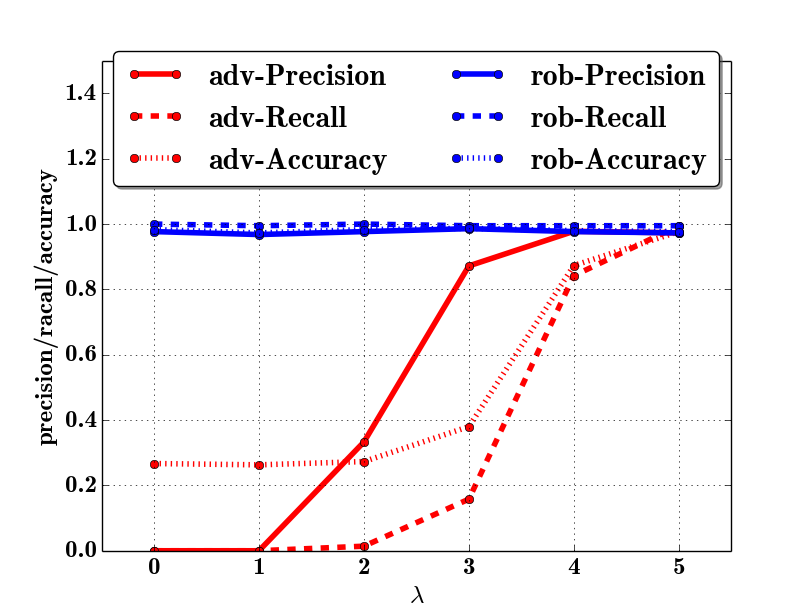} &
\includegraphics[scale=0.193]{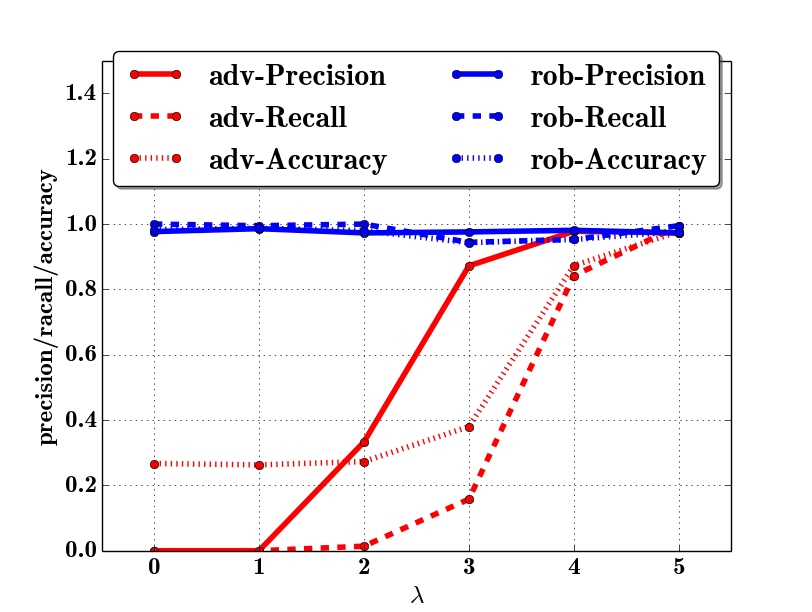} \\
(a) & (b) & (c) 
\end{tabular}
\caption{Performance of baseline (\emph{adv-}) and \emph{RAD} (\emph{rob-}) as a function of cost
  sensitivity $\lambda$ for Enron dataset with binary features testing on adversarial instances. (a) both defender and adversary use the L1 distance cost function, (b) both defender and adversary use the quadratic distance cost function, (c) adversary uses quadratic cost function while defender estimates it based on exponential cost.
  }
\label{fig:cost_variation}
\end{figure*}

\begin{figure*}[tbfh]
\centering 
\begin{tabular}{cccc}
\includegraphics[scale=0.23]{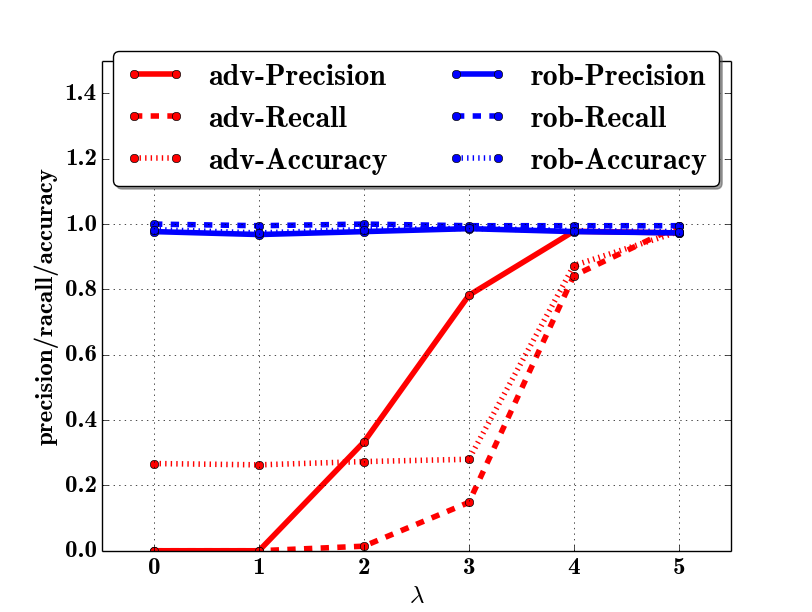} &
\includegraphics[scale=0.23]{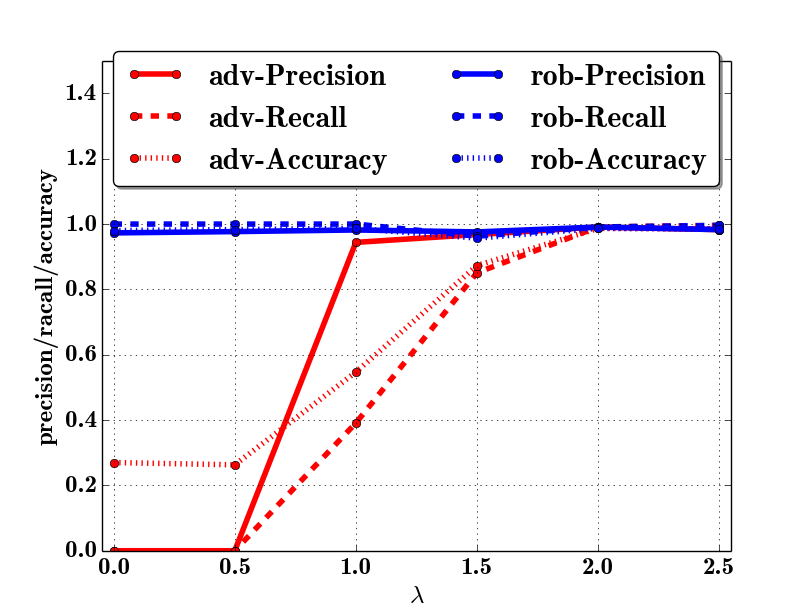} \\
(a) & (b) 
\end{tabular}
\caption{Performance of baseline (\emph{adv-}) and \emph{RAD} (\emph{rob-}) as a function of cost
  sensitivity $\lambda$ for Enron dataset with binary features testing on adversarial instances. (a) both defender and adversary use the equivalence-based cost function, (b) adversary uses equivalence-based cost function while defender estimates it based on exponential cost.
  }
\label{fig:cost_equivalence}
\end{figure*}

\subsection*{Experiments for Clustering Malicious Instances}
To efficiently speed up the proposed algorithm, here we cluster the malicious instances and use the
center of each cluster to generate the potential ``evasion" instances for the retraining framework.
Figure~\ref{fig:clustering_res} shows that the running time can be reduced by applying the clustering algorithm to the original malicious instances and the classification performance stays pretty stable for different learning models. 

\begin{figure*}[tbfh]
\centering 
\begin{tabular}{cccc}
\includegraphics[scale=0.23]{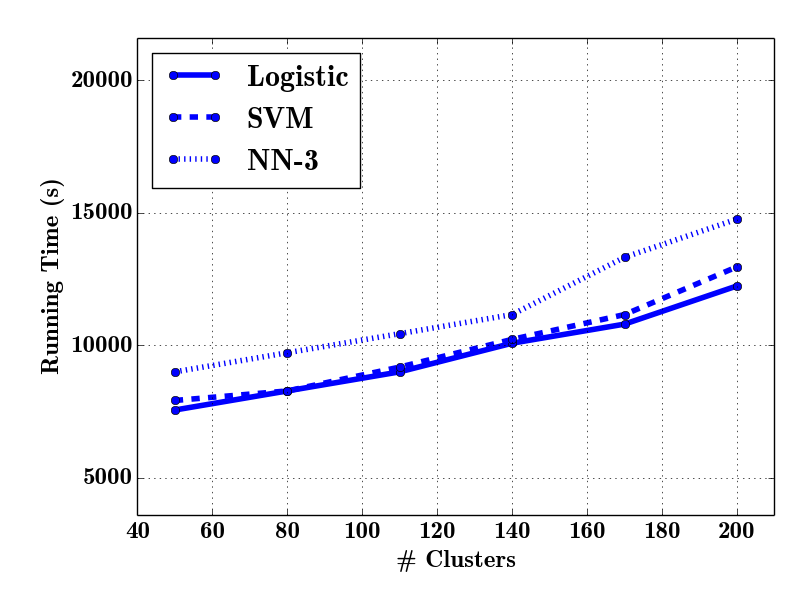} &
\includegraphics[scale=0.23]{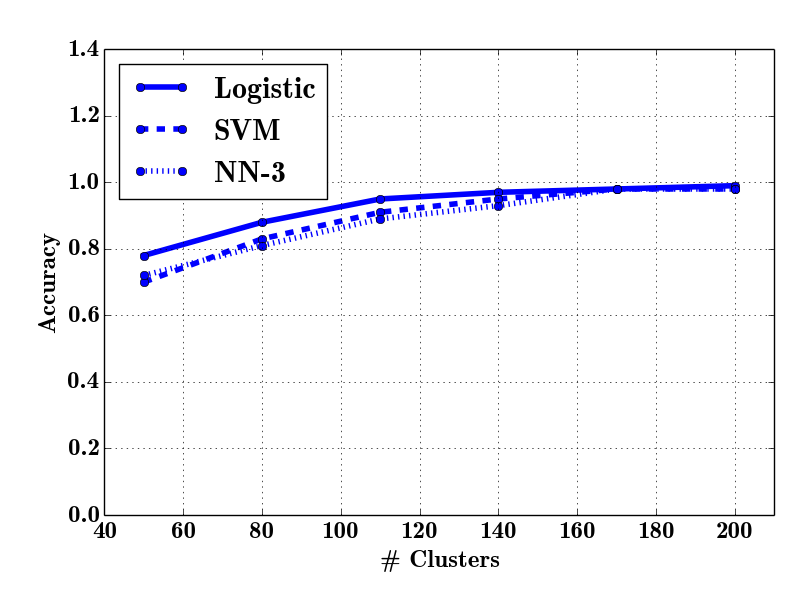} \\
(a) & (b) 
\end{tabular}
\caption{Performance of different learning models based on the number of clusters for Enron dataset testing on adversarial instances. (a) Running time, (b) classification accuracy of \emph{RAD}.
  }
\label{fig:clustering_res}
\end{figure*}

\end{document}